\documentclass[preprint]{elsarticle}

\usepackage{array,xspace,multirow,hhline,tikz,colortbl,tabularx,booktabs,fixltx2e,amsmath,amssymb,amsfonts,amsthm}
\usepackage{algorithm}
\usepackage{algorithmic}
\usepackage{verbatim,ifthen}
\usepackage{enumitem}
\usepackage{pifont}
\usepackage{ifthen}
\usepackage{calrsfs,mathrsfs}
\usepackage{bbding,pifont}
\usepackage{pgflibraryshapes}
\usepackage{lscape}

	\usepackage{varioref}

\definecolor{light-gray}{gray}{0.9}

	\newcommand{\AB}{\ensuremath{{A}}}

	\newtheorem{proposition}{Proposition}%
\newcommand{\av}[0]{\ensuremath{\mathit{AV}}\xspace}	
\newcommand{\sav}[0]{\ensuremath{\mathit{SAV}}\xspace}
\newcommand{\csav}[0]{\ensuremath{\mathit{CSAV}}\xspace}
\newcommand{\msav}[0]{\ensuremath{\mathit{MSAV}}\xspace}

	\usepackage{enumitem}
	\setenumerate[1]{label=\rm(\it{\roman{*}}\rm),ref=({\it\roman{*}}),leftmargin=*}
	\newlength{\wordlength}

\begin{document}

	\title{Algorithms for two variants of\\
	 Satisfaction Approval Voting}
	
	
		\author[nicta]{Haris Aziz} \ead{haris.aziz@nicta.com.au}
		\author[nicta]{Toby Walsh}  \ead{toby.walsh@nicta.com.au}
	 \address[nicta]{NICTA and UNSW, 223 Anzac Parade, Sydney, NSW 2033, Australia, \\Phone: +61 2 8306 0490}

\begin{abstract}
Multi-winner voting rules based on approval ballots have received increased attention in recent years. In particular Satisfaction Approval Voting (SAV) and its variants have been proposed. In this note, we show that the winning set can be determined in polynomial time for two prominent and natural variants of SAV.
We thank Arkadii Slinko 
for suggesting these problems in a talk at the Workshop on Challenges in Algorithmic Social Choice, Bad Belzig, October 11, 2014.
\end{abstract}

\maketitle

\section{Introduction}

Multi-winner voting rules based on approval ballots have received increased attention in recent years~\citep{ABC+14c,BKS07a,BMS06a,Kilg10a,MPRZ08a}. A natural way is \av in which the candidates with the highest number of approvals are selected. However \av is not very good at representing diverse intetests and other rules based on approval ballots have been devised.
In particular Satisfaction Approval Voting (\sav) and its variants have been proposed~\citep{BrKi14a,KiMa12a}. We show that the winning set can be computed in polynomial time for two prominent variants of \sav.

\section{Approval-based rules}

We consider the social choice setting  there is a
set of agents $N=\{1,\ldots, n\}$ and a set of candidates $C=\{c_1,\ldots, c_m\}$. Each
agent $i\in N$ expresses an approval
ballot $A_i\subset C$ that represents the subset of candidates that he
approves of, yielding a set of approval ballots $\AB = \{A_1,\ldots, A_n\}$.  \emph{Approval-based multi-winner rules} that take
as input $(C,\AB, k)$ and return the subset $W \subseteq C$ of size $k$ that is the winning set. 
Some rules may not return a winning set of a specific size and the size can vary from $1$ to $m$.
We consider the following rules. 

\paragraph{Satisfaction Approval Voting ($\sav$)}
\sav  finds a set $W \subseteq C$ of size $k$ that maximizes
$Sat(W)=\sum_{i\in N}\frac{|W\cap A_i|}{|A_i|}$.

\paragraph{Constrained Satisfaction Approval Voting ($\csav$)}
\csav  finds a set $W \subseteq C$  that maximizes
$CSat(W)=\sum_{i\in N}\frac{|W\cap A_i|}{|W|}$.

\paragraph{Modified Satisfaction Approval Voting ($\msav$)}
\msav  finds a set $W \subseteq C$ that maximizes
$MSat(W)=\sum_{i\in N}\frac{|W\cap A_i|}{\min(|A_i|,|W|)}$.\\

In contrast to \sav, \csav and \msav coincide with \av for the case of $k=1$. In this sense, they can be considered as proper generalization of \av for larger $k$.
$\csav$ and $\msav$ can also be used to select a winning set of a given size $k$ by selecting a set of size $k$ with the highest $CSat$ and $MSat$ scores respectvely.
Note that if the size of the winning set is not specified, then whereas $\sav$ gives the set of all candidates the highest score, $\csav$ and $\msav$ may not give the set of all candidates the highest score. We consider both the problem of computing the winning set of a particular size as well as the problem of computing the winning set of any size.

\section{Computational results}

\citet{AGG+14a} showed that \sav can be implemented in polynomial time. We show that \csav and \msav winning set can also be computed in polynomial time. 

\begin{proposition}
	There exists a polynomial-time algorithm to compute a set of size $k$ with the maximum CSat score.
	\end{proposition}
	\begin{proof}
		Note that for $|W|=k$,
		\[CSat(W)=\sum_{i\in N}\frac{|W\cap A_i|}{|W|}=\sum_{c\in W}\sum_{i\in N, c\in A_i}\frac{1}{k}\]
		where we say that $\sum_{i\in N, c\in A_i}\frac{1}{k}$ is the \emph{CSat score of an individual candidate with respect to set size $k$. }
Hence computing a winning set of size $k$ is equivalent to selecting $k$ candidates with the highest CSat scores with respect to set size $k$. Since the CSat score of a particular candidate can be computed in polynomial time, a set of size $k$ with the highest CSat score can be computed in polynomial time by selecting $k$ candidates with the highest CSat scores.\end{proof}

\begin{proposition}
	There exists a polynomial-time algorithm to compute a set of arbitrary size with the maximum CSat score.
	\end{proposition}
	\begin{proof}
		For $k=1$ to $|C|$, compute a $k$-set with the maximum CSat score. Identify the $k$ for which the $k$-set with the maximum CSat score has the maximum CSat score among all $k$s. Return such a $k$-set.
	\end{proof}

\begin{proposition}
	There exists a polynomial-time algorithm to compute a set of size $k$ with the maximum MSat score.
	\end{proposition}
	\begin{proof}
		Note that for $|W|=k$,
		\[MSat(W)=\sum_{i\in N}\frac{|W\cap A_i|}{\min(|A_i|,|W|)}=
		\sum_{i\in N}\frac{|W\cap A_i|}{\min(|A_i|,k)}=
		\sum_{c\in W}\sum_{i\in N, c\in A_i}\frac{1}{\min(|A_i|,k)}\]
		where we say that $\sum_{i\in N, c\in A_i}\frac{1}{\min(|A_i|,k)}$ is the \emph{MSat score of an individual candidate with respect to set size $k$.} 
Hence computing a winning set of size $k$ is equivalent to selecting $k$ candidates with the highest MSat scores with respect to set size $k$. Since the MSat score of a particular candidate can be computed in polynomial time, a set of size $k$ with the highest MSat score can be computed in polynomial time by selecting $k$ candidates with the highest MSat scores.\end{proof}

\begin{proposition}
	There exists a polynomial-time algorithm to compute a set of arbitrary size with the maximum MSat score.
	\end{proposition}
	\begin{proof}
		For $k=1$ to $|C|$, compute a $k$-set with the maximum MSat score. Identify the $k$ for which the $k$-set with the maximum MSat score has the maximum MSat score among all $k$s. Return such a $k$-set.
	\end{proof}

	We see that just like the \sav, two prominent variants of \sav can computed in polynomial time.

%

%

\end{document}